\documentclass[lettersize,journal]{IEEEtran}
\usepackage[table]{xcolor}
\usepackage{colortbl}
\usepackage{float}
\usepackage{array}
\usepackage{graphicx} 
\usepackage{algorithm}
\usepackage{algpseudocode}
\usepackage{amsmath}
\usepackage{pgfplots}
\usepackage{pgfplotstable}
\usepackage{paralist}
\usepackage{multirow}
\usepackage{tcolorbox}
\usepackage{pifont} 
\usepackage{tabularx}
\usepackage{amssymb}
\usepackage{booktabs}
\usepackage{enumitem}
\usepackage{tikz} 
\usepackage{makecell}
\usepackage{bm}
\usepackage{bbding}
\usepackage{amsthm}
\usepackage{soul}
\usepackage{amsmath}
\usepackage{amssymb}
\usepackage{hyperref}
\usepackage{url} 
\definecolor{lightred}{HTML}{EEA599}
\definecolor{lightyellow}{HTML}{FFE9BE}
\definecolor{lightgreen}{HTML}{E3EDE0}
\newtheorem{theorem}{Theorem}[section]

\begin{document}

\title{Empowering IoT Firmware Secure Update with Customization Rights}

\author{
Weihao Chen, Yansong Gao, Boyu Kuang, Jin~B.~Hong, Yuqing Zhang, and Anmin Fu%
\thanks{Weihao Chen and Boyu Kuang are with the School of Cyber Science and Engineering, Nanjing University of Science and Technology, Nanjing, China (e-mail: chenweihao0316@njust.edu.cn; kuang@njust.edu.cn).}%
\thanks{Yansong Gao and Jin B. Hong are with the School of Computer Science and Software Engineering, The University of Western Australia, Perth, Australia (e-mail: gao.yansong@hotmail.com; jin.hong@uwa.edu.au).}%
\thanks{Yuqing Zhang is with the School of Cyber Security, University of Chinese Academy of Sciences, Beijing, China (e-mail: zhangyq@ucas.ac.cn).}%
\thanks{Anmin Fu is with the School of Cyber Science and Engineering, Nanjing University of Science and Technology, Nanjing, China (e-mail: fam\_0522@163.com).}%
}

\maketitle

\begin{abstract}
Firmware updates remain the primary line of defence for IoT devices. However, the update channel itself has become a well-established attack vector. Existing defences predominantly focus on securing monolithic firmware images, leaving module-level customisation—a rising demand among users—largely unprotected and insufficiently explored.

To bridge this gap, we conduct a pilot study on the update workflows of 200 Linux-based IoT devices across 23 vendors. Our analysis uncovers five previously undocumented vulnerabilities stemming from customisation practices. A broader investigation of update-related CVEs from 2020 to 2024 further reveals that over half originate from customisation-induced issues. These findings highlight a critical yet underexamined reality: as customisation increases, so does the attack surface—yet current defences fail to keep pace.
We thus present \textbf{IMUP} (\textit{Integrity-Centric Modular Update Platform}), the first framework to jointly address the two key challenges identified in our pilot study: (C1) constructing a trustworthy cross-module integrity chain, and (C2) scaling update performance under mass customisation. 
IMUP integrates three core techniques: (i) \textit{per-module chameleon hashing} to enable verifiable integrity across customised firmware compositions, (ii) \textit{server-side proof-of-work offloading} to preserve the lightweight nature of IoT devices, and (iii) \textit{server-side caching} to reuse previously generated module combinations, significantly reducing the need to rebuild customised images from scratch.

Security analysis demonstrates that even under extreme conditions—where up to 95\% of secret keys are exposed—forging a valid image still incurs a computational cost over $300\times$ greater than that of the legitimate server.
Extensive experiments on heterogeneous IoT devices and a server cluster validate IMUP’s efficiency: for a single wave of 30{,}000 customised update requests, server-side generation time is reduced by $2.9\times$, and device downtime is reduced by $5.9\times$, compared to the package-manager baseline.
All code, datasets, and reproduction scripts are available at
\href{https://github.com/Integrity-Centric}{https://github.com/Integrity-Centric}.
\end{abstract}


\section{Introduction}

\IEEEPARstart{F}{irmware} updates for Internet‑of‑Things (IoT) devices enhance security through vulnerability patching and improve user experience by introducing new features~\cite{Mavromatis2022ReliableIF}. However, modern IoT devices increasingly encompass a wide range of features, and each user or device may require only a subset of features. For instance, a smart camera might need motion detection without cloud-based recognition, whereas a router could benefit from VPN support without ad-blocking. This requires the customization of IoT firmware updates to selectively update only the necessary functionalities or features. However, we note that there does not exist a widely acknowledged definition of the IoT firmware update customization.
Drawing on Mahurkar \textit{et al.}\text{’s} concept of software customization~\cite{10.1145/3590140.3629121}, we define Customization Rights as the capability to selectively update only the necessary IoT firmware modules—rather than applying a monolithic image—within an authorized and integrity-preserving process. This notion emphasizes controlled flexibility: users can tailor updates to specific functional requirements, while the system maintains a trustworthy update boundary.
Such customized firmware updates present a unique challenge, as the system must enforce per-module authentication and maintain cross-module integrity (e.g., if an update image is injected with an untrusted functional module or misses a required dependency module, the entire update should be rendered invalid)—a significantly more complex task than verifying a single monolithic firmware.

\begin{figure}[htbp]
  \centering
  \includegraphics[width=\columnwidth]{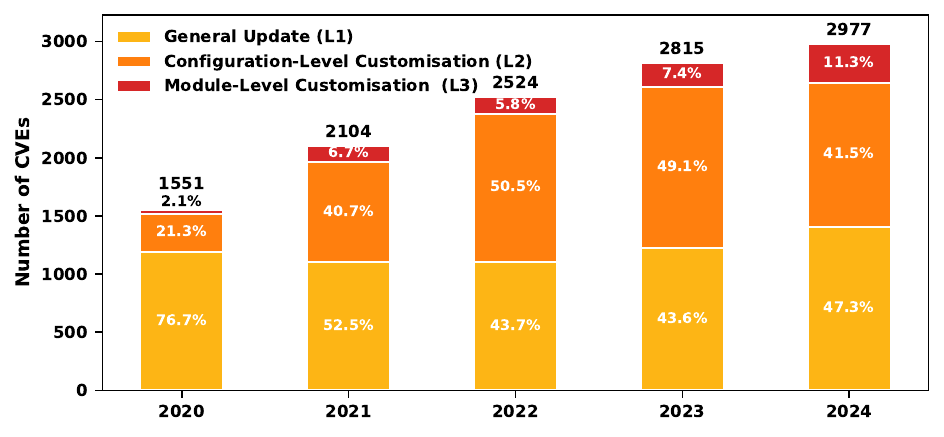}
  \caption{Firmware update vulnerabilities by customisation (2020-2024)}
  \label{fig:firmware-cve}
\end{figure}

\noindent{\bf From Monolithic Update to Module Customization.} Although firmware updates are widely considered a key method for mitigating vulnerabilities, real-world cases repeatedly expose the high risks within the update chain itself. For example, ASUS ShadowHammer abused a valid certificate to distribute tampered firmware~\cite{shadowhammer_asus_2019}; the CrowdStrike–Microsoft incident caused global outages due to improper update handling~\cite{tipranks2024}; and Wu \textit{et al.} (USENIX Security 2024) systematically revealed multiple scenarios in which signed patches could still be manipulated~\cite{wu2024your}. The OWASP IoT Top 10 also lists “insecure update mechanism” as one of the top five threats to embedded devices~\cite{owasp_iot_top_10_2018}. However, most of these works assume that firmware is released as a monolithic package and verified in a single step, ignoring the increasingly common practice of module-level customisation: manufacturers or users may replace individual function modules, inject scripts, or overwrite configurations. These actions fragment the verification boundary and expand the attack surface, yet their scale and evolution have not been systematically quantified. Based on this gap, we raise two early research questions:

\begin{itemize}
  \item \textbf{RQ1:} To what extent does module-level firmware customisation expand the attack surface of IoT devices?
\end{itemize}

\vspace{0.10cm}

\noindent{\bf Pilot Study Overview.} To address RQ1, we conducted an empirical investigation involving over 200 devices across 23 vendors, complemented by reverse engineering to construct a unified threat model (Section~\ref{ThreatModel}).  
We systematically reviewed each vendor’s customisation support strategy and identified three core attack surface touchpoints introduced by modular update mechanisms.  
Through real-device testing, we discovered five previously unknown 0-day vulnerabilities (CVE-1 to CVE-5), which informed our characterization of Layer-2 (L2) and Layer-3 (L3) customisation-induced vulnerabilities.  
Leveraging these characteristics as labels, we performed an automated classification of all update-related CVEs disclosed from 2020 to 2024. Our analysis reveals that fraction of module-level (L3) vulnerabilities have surged nearly $5\times$ over five years (Figure~\ref{fig:firmware-cve}).  
These findings collectively demonstrate that: (i) customisation significantly enlarges the firmware update attack surface, and (ii) its security implications are rapidly escalating yet remain insufficiently addressed.

\noindent{\bf Lack of Existing Solution.} 
While some vendors advertise “customization features,” their incomplete solutions often shift the burden—and associated risks—of customisation onto end users~\cite{openwrtAbout}, thereby expanding the attack surface while neglecting legitimate demands for flexibility. 
Our pilot study affirms that over 50\% of update-related CVEs reported between 2020 and 2024 originate from these customisation interfaces, underscoring the urgency of the problem.  
Although hardware roots of trust and digital signatures are effective at securing monolithic firmware images~\cite{ebbers2022large}, they offer no support for partial, module-level updates. 

\noindent\textbf{Device-Side Baselines.} Existing firmware update mechanisms predominantly follow one of two baselines:  
\textbf{Monolithic Rebuild.} Any modification from end-devices triggers a full image rebuild and re-signing, incurring high build, bandwidth, and device downtime costs.  
\textbf{Package Manager (e.g., Opkg).} Allows incremental installation but lacks robust isolation and a global integrity chain, exposing the system to threats such as dependency confusion and rollback attacks.  
Neither approach adequately balances strong security guarantees with the flexibility required for scalable customisation.

{\noindent\textbf{Server-Side Optimisations.} Extensive work on firmware distribution has explored incremental updates, differential patching, and repository-based installs to cut bandwidth~\cite{10.1145/3384419.3430471,liu2023light}.  
However, these optimisations still assume a single, vendor-defined image.  
Once end users are free to mix and match modules, the search space of valid images grows exponentially.  
A single critical patch can therefore fan out into thousands of distinct requests, multiplying certificate checks and I/O on the server.

Together, the lack of cross-module integrity guarantees and diverse customization requests render two challenges:

\begin{itemize}[leftmargin=*]
  \item \textbf{C1 (Security):} How can we construct a trustworthy cross-module integrity chain without negating user customization?
  \item \textbf{C2 (Scalability):} How can resource-constrained IoT devices and large-scale servers sustain the computational and storage overhead induced by massive user-driven customization?
\end{itemize}

\vspace{0.10cm}
\noindent\textbf{Our Solution.}
We propose \textbf{IMUP}, a firmware update mechanism that reconciles security, efficiency, and customisation—three goals that existing solutions struggle to achieve simultaneously.  
Unlike prior approaches, IMUP embeds user-defined customisation directly into its trusted computing base.  
To enforce per-module integrity (\textbf{C1}), IMUP employs a chameleon hash with weak collision resistance, enabling verifiable integrity while supporting customisable components.  
Heavyweight cryptographic operations—such as collision generation and validation—are offloaded to a resource-rich server via Proof-of-Work offloading, allowing constrained devices to perform only lightweight tag verification.  
To address performance bottlenecks (\textbf{C2}), IMUP features a modular update pipeline where clients retrieve only the required components. On the server side, previously built module combinations are cached and reused across requests, significantly reducing rebuild and signing costs. Our contributions are threefold.

\begin{itemize}[leftmargin=*]\itemsep0pt
\item We conduct the first large-scale analysis of customisation practices across 200 devices from 23 vendors, uncovering five unknown CVEs. Our findings expose systemic risks induced by customisation-enabled firmware updates.
\item We present the first framework that integrates module-level chameleon hashing with proof-of-work offloading and server-side cache reuse to jointly ensure update integrity and scalable performance.
\item Evaluations on real-world devices and a server show that IMUP reduces firmware generation latency by $2.92\times$ and device downtime by $5.93\times$ compared to a package-manager baseline. Even under $95\%$ key exposure, forging a valid image remains $300\times$ expensive for an attacker than legitimate server-side updates—ensuring resilient disaster recovery.
\end{itemize}

\section{Related Work}

\subsection{Security-Oriented Approaches}
A substantial body of research focuses on ensuring the integrity and authenticity of firmware during updates. These approaches commonly employ cryptographic primitives for signing and verification, or leverage trusted execution environments to enforce authenticity. For example, Xie \textit{et al.}~\cite{10422815} propose a lightweight authenticated key agreement scheme suitable for resource-constrained IoT devices, while Yu \textit{et al.}~\cite{10583906} develop EDASVIC for dynamic integrity checks in industrial cloud platforms. Blockchain-based solutions have also been explored, as in Hu \textit{et al.}~\cite{hu2019autonomous}, which employs immutable ledgers and batch verification to secure over-the-air updates.}

While these mechanisms provide strong protection against tampering and unauthorized updates, they uniformly assume a monolithic firmware image, where integrity is enforced globally via a single signature or hash. Such designs do not consider modular update workflows and therefore fail to address cross-module dependency risks introduced by customization.

\subsection{Efficiency-Focused Solutions}

To mitigate bandwidth and energy overheads during large-scale deployments, researchers have investigated mechanisms that optimize data transfer and reduce computational costs. Incremental and differential update techniques, such as those proposed by Arakadakis \textit{et al.}~\cite{10.1145/3384419.3430471}, reduce update size by transmitting only changed segments. Liu \textit{et al.}~\cite{liu2023light} improve energy efficiency by introducing lightweight flash write operations for energy-harvesting devices. Network-level optimizations, like the CoAP-based two-phase dissemination in Park \textit{et al.}~\cite{park2019two}, aim to accelerate distribution in resource-constrained environments. Blockchain-based methods~\cite{hu2019autonomous,10153598} have further integrated batch verification and distributed trust to reduce centralized bottlenecks.

Although these techniques significantly improve scalability and responsiveness, they largely ignore integrity challenges under modular customization. Their efficiency models presuppose that the update is a single, fixed image, leaving the verification of interdependent modules unexplored.

\subsection{Customization-Oriented Frameworks}

Recent efforts have begun to explore flexible update frameworks tailored for IoT devices. UpKit~\cite{8884933}, for instance, provides an open-source, portable mechanism for firmware updates, enabling lightweight integration of modular components. Similarly, SecuCode~\cite{su2019secucode} addresses secure wireless dissemination for highly constrained devices, allowing selective code deployment in dynamic environments. These works demonstrate the feasibility and utility of customized updates, where functionality can be adapted without a full image rewrite.

However, existing frameworks prioritize flexibility over integrity. They generally lack mechanisms to enforce chained verification across modules, making them susceptible to rollback attacks, dependency manipulation, and malicious module injection. As such, they inadvertently introduce new attack vectors by decoupling update logic from global integrity assurance.


\subsection{Existing Limitations and Challenges}

While prior research has addressed specific aspects of firmware update security, efficiency, and flexibility, none provides a unified solution that balances these dimensions under large-scale customization. Our analysis reveals three critical gaps:

\begin{enumerate}[label=(\arabic*)]
    \item \textbf{Assumption of Monolithic Firmware:} Security mechanisms~\cite{10422815,10583906,hu2019autonomous} rely on global signatures for entire images, which fail when firmware is decomposed into independently updatable modules.
    
    \item \textbf{Lack of Cross-Module Integrity:} Efficiency-driven approaches~\cite{10.1145/3384419.3430471,liu2023light,park2019two,hu2019autonomous,10153598} optimize resource consumption but ignore integrity dependencies introduced by modular customization, making it possible for attackers to reorder or inject modules.
    
    \item \textbf{Customization Without Safety Guarantees:} Frameworks such as UpKit~\cite{8884933} and SecuCode~\cite{su2019secucode} prioritize functional flexibility but omit chained verification across modules, exposing vulnerabilities like rollback attacks, dependency manipulation, and malicious injection.
\end{enumerate}

These limitations underscore the need for an integrity-centric solution that supports granular customization without sacrificing security or scalability—motivating the design of IMUP presented in Section~\ref{Design}.

\subsection{Comparison of Related Works}

To position IMUP within the context of existing research, we compare prior works across six evaluation dimensions that collectively capture update security, performance, and flexibility. These include:

\begin{itemize}
    \item \textbf{Update Security:} Whether the scheme enforces cryptographic integrity and authenticity checks on firmware images.
    \item \textbf{Update Efficiency:} Indicates mechanisms to reduce update size or verification overhead, such as incremental transfer or batch verification.
    \item \textbf{Wireless Update:} Denotes support for over-the-air (OTA) update delivery.
    \item \textbf{Server Cost:} Reflects computational overhead on the update server; a low-cost design minimizes dependency on centralized verification or heavy cryptographic operations.
    \item \textbf{Config-Level Customization (L2):} Ability to modify configuration parameters (e.g., network scripts) without rebuilding the entire image.
    \item \textbf{Module-Level Customization (L3):} Ability to add or remove functional modules during an update while preserving integrity guarantees.
\end{itemize}

\begin{table}[htbp]
\belowrulesep=0pt
\aboverulesep=0pt
\centering
\caption{Comparison Between Related Works}
\resizebox{\linewidth}{!}{
\begin{tabular}{
>{\arraybackslash}c
>{\arraybackslash}c
>{\arraybackslash}c
>{\arraybackslash}c
>{\arraybackslash}c
>{\arraybackslash}c
>{\arraybackslash}c
}
\toprule
Scheme & \makecell{Update \\ Security} & \makecell{Update \\ Efficiency} & \makecell{Wireless \\ Update} & \makecell{Server \\ Cost} & \makecell{Config-Level \\ (L2)} & \makecell{Module-Level \\ (L3)} \\
\midrule
IncUpd\cite{10.1145/3384419.3430471} & \ding{56} & \ding{52} & \ding{56} & \ding{56} & \ding{56} & \ding{56} \\
LFW\cite{liu2023light} & \ding{56} & \ding{52} & \ding{56} & \ding{56} & \ding{56} & \ding{56} \\
TP-FOTA\cite{park2019two} & \ding{56} & \ding{52} & \ding{52} & \ding{52} & \ding{56} & \ding{56} \\
BatchAuth\cite{10704696} & \ding{52} & \ding{52} & \ding{56} & \ding{56} & \ding{56} & \ding{56} \\
UpKit\cite{8884933} & \ding{56} & \ding{52} & \ding{52} & \ding{52} & \ding{52} & \ding{56} \\
SecuCode\cite{su2019secucode} & \ding{52} & \ding{52} & \ding{52} & \ding{56} & \ding{56} & \ding{52} \\
\textbf{IMUP (Ours)} & \ding{52} & \ding{52} & \ding{52} & \ding{52} & \ding{52} & \ding{52} \\
\bottomrule
\end{tabular}
}
\label{table:FirmwareUpdateComparison}
\end{table}

Table~\ref{table:FirmwareUpdateComparison} summarizes how representative schemes align with these criteria. Notably, IMUP is the only approach that satisfies both L2 and L3 customization while maintaining cross-module integrity and efficiency.

\section{Pilot Study} \label{PilotStudy}

Existing studies have yet to clearly demonstrate the security threats introduced by customisation in the IoT firmware update process. To bridge this gap, we begin by systematically reviewing and comparing the customisation practices of major firmware platforms, uncovering several previously undocumented security issues. Building on the shared characteristics of these issues, we further label and classify all public update-related vulnerabilities reported from 2020 to 2024. This analysis enables us to quantify the evolution of customisation-induced risks over the past five years.



\subsection{Customization Overview}

 Device customization has become a key factor in ensuring the efficient operation of IoT devices across diverse and dynamic settings \cite{9355854}. Mehdipour \textit{et al.} \cite{9355854} highlight that the diversity of device functionalities, dynamic network topologies, and evolving security threats demand customization requirements that standard solutions cannot adequately address. Additionally, Sousa \textit{et al.} \cite{298092,298238} emphasize that customization allows IoT devices to achieve higher security and more efficient communication tailored to specific operational environments. User customization has also emerged as a critical factor influencing the choice of IoT devices \cite{casaca2024}. Personalization is vital for improving user satisfaction and experience \cite{7881335,9326414}, as it enables devices to adapt to individual preferences and requirements.

\begin{table}[htbp]
\centering
\scriptsize  
\caption{Summary of Partial Firmware Manufacturers' Update Strategies (full summary available on GitHub)}
\begin{tabular}{|>{\centering\arraybackslash}m{0.15\linewidth}
                |>{\centering\arraybackslash}m{0.13\linewidth}
                |>{\centering\arraybackslash}m{0.18\linewidth}
                |>{\centering\arraybackslash}m{0.15\linewidth}
                |>{\centering\arraybackslash}m{0.13\linewidth}|}
\hline
Classification & IoT Vendor & Key Protections & Update Frequency & Package Manager \\
\hline
\cellcolor{lightgreen} \textbf{Open Source} & OpenWrt & \cellcolor{white}Rollback: \textcolor{red}{No} \newline Auto: \textcolor{red}{No} & \cellcolor{lightgreen}Monthly & \cellcolor{lightgreen}opkg \\
\hline
\cellcolor{lightred} \textbf{Commercial} & Cisco & \cellcolor{white}Rollback: \textcolor{green!50!black}{Yes} \newline Auto: \textcolor{green!50!black}{Yes} & \cellcolor{lightgreen}Monthly & \cellcolor{lightred}No \\
\cline{2-5}
\cellcolor{lightred} & TP-Link & \cellcolor{white}Rollback: \textcolor{green!50!black}{Yes} \newline Auto: \textcolor{green!50!black}{Yes} & \cellcolor{lightyellow}3--6 months & \cellcolor{lightred}No \\
\hline
\end{tabular}
\label{tab:firmware-update-strategies-partial}
\end{table}

To assess whether current firmware updates adequately address customization needs, we conducted an empirical investigation into manufacturers' customization practices and analyzed the update management strategies of leading firmware providers. In the IoT ecosystem, firmware providers are primarily categorized into open-source and commercial vendors. Open-source providers, such as OpenWrt and Arduino, offer firmware that the community can freely use and modify. In contrast, commercial providers supply proprietary firmware developed and maintained by device manufacturers or specialized vendors.

To systematically assess these manufacturers' firmware update practices, we have reviewed official documentation, product manuals, technical support websites, and publicly available technical reports. We establish four key indicators for evaluation: prevention of firmware rollback, support for automatic updates, update frequency, and the package manager utilized. These indicators effectively capture the diverse strategies and measures employed by manufacturers during the firmware upgrade process. 
For this study, we collected information on 200 IoT devices\footnote{Detailed at \href{https://github.com/Integrity-Centric/IMUP_Firmware_DataSet}{https://github.com/Integrity-Centric/IMUP\_Firmware\_DataSet}}, and an analysis of firmware update strategies among mainstream manufacturers reveals stark differences between open-source and commercial firmware providers. The key findings on representative vendors are summarized in Table \ref{tab:firmware-update-strategies-partial}, while the complete table of all analyzed vendors is available in \href{https://github.com/Integrity-Centric/IMUP_Firmware_DataSet}{GitHub}.

\begin{enumerate}
    \item Open-source projects such as OpenWrt emphasize flexibility and extensive user customization, allowing users to modify firmware functionalities to meet specific needs. However, these projects typically do not support automatic updates or firmware rollback protection, resulting in vulnerabilities within their security measures and insufficient security assurances. 
    \item In contrast, commercial firmware providers like Cisco, TP-Link, and Huawei prioritize enhancing device management efficiency through automation and robust security mechanisms. These vendors generally support automatic updates and firmware rollback protection, implement stringent version management, and maintain frequent update cycles to bolster device security.
\end{enumerate}

In summary, the analysis of firmware update strategies shows that while open-source firmware projects prioritize user customization and flexibility, they often lack robust security features like automatic updates and rollback protection, which limits their widespread adoption. On the other hand, commercial firmware providers offer more automated, secure update mechanisms. Nevertheless, they sacrifice customization to obtain stronger automation, highlighting that widespread adoption of extensive firmware customization remains limited. 

\subsection{Unified Threat Model for Updates} \label{ThreatModel}

\begin{figure}[htbp]
    \centering
    \includegraphics[width=0.48\textwidth]{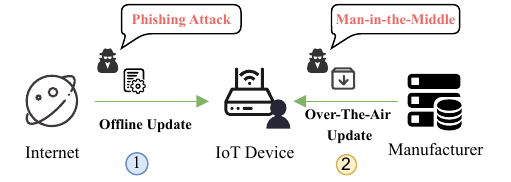}
    \caption{Offline vs.\ OTA update paths and two customisation-specific risk points:\textcircled{\scriptsize 1} L2 (config-level tampering); \textcircled{\scriptsize 2} L3 (module-level forging).}
    \label{Threat Models}
\end{figure}

\noindent\textbf{Update modes.}
• \emph{Offline} (risk \textcircled{\scriptsize 1}): users sideload configuration files or scripts that bypass signing.
• \emph{OTA} (risk \textcircled{\scriptsize 2}): vendor-signed image accepts extra modules whose cross-module
integrity is unchecked.

\noindent\textbf{Adversary.}
We assume remote attackers capable of MITM and binary tampering; physical attacks are out of scope (details in Appx ~\ref{Appendix:Threat-Model}).

\subsection{Identified Categories of Customization Vulnerabilities} \label{Vulnerabilities}

During our testing of the firmware update chains from 23 vendors and over 200 devices, we discovered five previously unknown 0-day vulnerabilities (referred to as CVE-1 to CVE-5\footnote{For {\it anonymity purposes}, we do not use the real CVE number here.})).
For completeness, we keep a baseline class L1—vulnerabilities that reside entirely inside the update process itself (e.g., a stored XSS in the web-based updater that never touches the image) and do not alter the core image boundary. Based on how customisation operations cross the image boundary, the five 0-days fall into two customization-dependent classes: L2 (configuration-level) and L3 (module-level).

\noindent \textbf{CVE-1 — L2: Configuration-Level Customisation Vulnerability.} A popular IP camera offers a feature for custom scripts. However, in its firmware update package, the script file is treated only as data and is not verified as strictly as the core image. An attacker can insert malicious code into the script filename. If a user downloads and applies this unverified file during an offline update, the device may be compromised.  
The common feature of L2 vulnerabilities is that \textbf{the main firmware image remains unchanged}, and the risk comes from weak verification of auxiliary files such as scripts, or rollback markers.

\noindent \textbf{CVE-5 — L3: Module-Level Customisation Vulnerability}
A certain vendor's router is designed to support firmware customization. However, its update process is vulnerable to a stack overflow attack. Although the official firmware package does not trigger this issue, an attacker can exploit the customization feature by modifying the length field in the firmware header. This allows a malicious firmware to bypass the integrity check and take control of the device. This case highlights the security risks caused by weak verification in customization-friendly systems.  
L3 vulnerabilities are characterized by \textbf{adding or replacing binary components}, which expands the attack surface to include the module loader, dependency resolution, and cross-module integrity chains.

\subsection{Quantify Customisation CVEs in Firmware Updates}

\noindent\textbf{Data Source and Method.} We used the \textbf{NVD v2 API}~\cite{nvd_api} to extract all CVE records related to ``update'' from 2020 to 2024. By using multiple queries with keywords such as ``\verb|firmware update|'', and applying a 120-day rolling window to avoid API rate limits, we obtained 14,389 candidate CVEs. Each CVE description was processed using the \textbf{L1/L2/L3 regular expression classifiers} proposed in the previous section.

Figure\,\ref{fig:firmware-cve} shows the overall distribution of the three vulnerability types from 2020 to 2024. Starting from 2022, \textbf{L2} vulnerabilities account for 50.5\%, \textbf{L3} for 5.8\%, and the remaining 43.7\% are L1, which only involve the update process itself. In other words, since 2022, \emph{more than half} of all update-related vulnerabilities are strongly related to customisation, highlighting its impact on the integrity chain. L3 vulnerabilities have grown the fastest: the number of module-level issues increased from only 32 in 2020 to 335 in 2024—an increase of nearly 857\% over five years.

\subsection{Key Finding}

Customization is a critical factor in enhancing the functionality, adaptability, and user satisfaction of IoT devices, particularly in diverse and dynamic environments, such as security cameras with privacy modes or sensors operating in fluctuating network conditions. However, our pilot investigation reveals that the current level of customization is insufficient, with limited support from manufacturers and a lack of standardized practices. This gap not only restricts user flexibility but also introduces significant security challenges, as customization expands the attack surface and facilitates adversary exploitation.

\textbf{Further quantitative analysis shows that the risks from module-level customisation (L3) are rising rapidly and have long been overlooked.} Our statistics on all firmware update-related \textsc{CVE}s from 2020 to 2024 show that vulnerabilities involving the \emph{addition or replacement of binary components} in firmware have increased nearly $5\times$ over five years.  
At the same time, existing industry solutions still follow the traditional “one-time signing of full images” model. They fail to build module-level integrity chains for third-party plugins, optional feature packages, or Stock Keeping Unit variants, resulting in growing blind spots in verification.

Based on the key findings of our pilot study, the rest of this work provides a unique perspective and solution on the interplay between customization and security, aiming to offer more robust and adaptable firmware update strategies.

\section{Preliminary}

\subsection{Chameleon Hash}
The Chameleon Hash Function (C\textsubscript{hash}) is a cryptographic primitive that supports trapdoor-enabled collision generation, originally introduced by Krawczyk \textit{et al.}~\cite{krawczyk1998chameleon}. In contrast to traditional hash functions, which are strictly collision-resistant, Chameleon Hash allows an authorized party with a trapdoor key to generate collisions in a controlled manner.

In this work, we adopt the identity-based Chameleon Hash construction proposed by Chen \textit{et al.}~\cite{chen2004chameleon}, which enhances key management by preventing trapdoor key exposure. This construction plays a critical role in our update mechanism: it enables legitimate firmware providers to modify update content—such as customized data—without changing the overall digest, thereby achieving efficient customization with integrity guarantees.
The Chameleon Hash is defined as:
\[
\textsf{C}_{\text{hash}}(m, r) = \textsf{C}_{\text{hash}}(m', r'),
\]
where \( m \neq m' \), and the collision-resistance property is selectively bypassed using a trapdoor function.

The value \( r' \) required to produce a collision with a new message \( m' \) is computed by the collision-finding algorithm:
\[
r' = \textsf{C}_{\rm sk}(m, m', r),
\]
where \( sk \) is the secret trapdoor key.

In our scheme, each pair \( \{m, r\} \) and \( \{m', r'\} \) forms a valid collision under the same digest. This property allows controlled content rewriting in firmware updates while preventing unauthorized manipulation, as only parties holding the trapdoor can generate such collisions.

\begin{figure*}[htbp]
    \centering
    \includegraphics[width=1\textwidth]{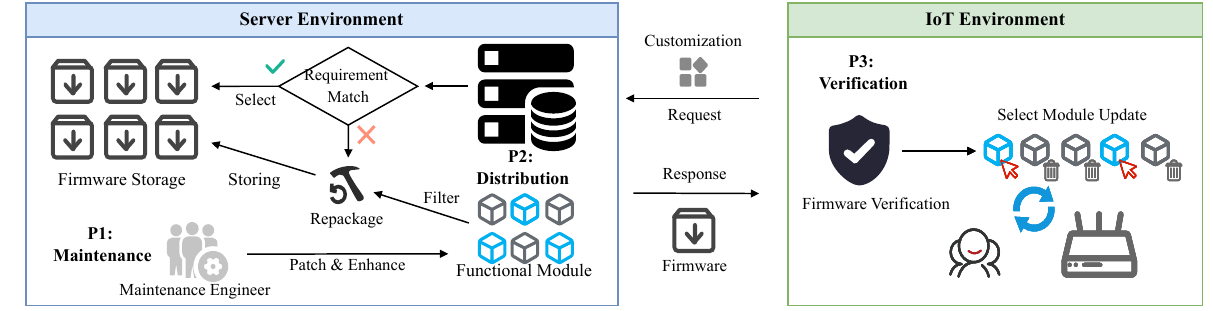}
    \caption{IMUP overview}
    \label{fig:Overview}
\end{figure*}

\subsection{Proof of Work (PoW)}

PoW is a cryptographic mechanism originally proposed by Dwork and Naor~\cite{dwork1992pricing} to introduce computational costs that mitigate abuse and prevent Denial-of-Service (DoS) attacks. IMUP leverages PoW to create computational asymmetry: resource-rich servers can efficiently solve PoW challenges, but adversaries attempting to forge firmware must incur significant computational expenses, thereby increasing their attack cost.

We formally define the PoW generation process as:
\[
\textsf{PoW}(T, m, D_{\text{PoW}}) \rightarrow (\textsf{SolutionNonce}, B_{\text{hash}}),
\]
where: \( T \) denotes the task structure (in IMUP, this involves repeated Chameleon Hash computations), \( m \) is the input message (e.g., customized metadata digest), \( D_{\text{PoW}} \) is the difficulty threshold, \( \textsf{SolutionNonce} \) is the computed nonce satisfying the difficulty, \( B_{\text{hash}} \) is the resulting hash value that meets the required threshold.

In our design, the task \( T \) involves repeatedly computing the Chameleon Hash over \( m \) with different nonce candidates until the resulting hash satisfies the difficulty constraint \( D_{\text{PoW}} \).

The corresponding PoW verification process is defined as:
\[
\textsf{Verify}(\textsf{SolutionNonce}, m, D_{\text{PoW}}) \rightarrow \text{Bool},
\]
where \( \text{Bool} \in \{0, 1\} \). The verifier recomputes the hash from \( (m, \textsf{SolutionNonce}) \) and checks whether it satisfies the specified difficulty. A result of 1 indicates successful verification; otherwise, failed.

This mechanism enhances the integrity of firmware updates by ensuring that only parties capable of performing sufficient computational work can generate valid proofs, while minimizing the burden on resource-constrained IoT devices.

\subsection{Notations}

The notations used in the rest of the paper, following that in \cite{chen2004chameleon}, are described below. In general, bold characters are used to denote vectors; for instance, a Chameleon Hash digest vector is represented as \(\mathbf{CDigest}\). Calligraphic font is used to denote sets; for example, a set of such digest vectors is represented as $\mathcal CDigest$, where \(\mathbf{CDigest}\) \(\in\) $\mathcal CDigest$.

\begin{itemize}[leftmargin=*]

    \item The Crypto Module (\(\mathbf{CModule}\)) is a fundamental cryptographic component of IMUP, and its generation algorithm is detailed in Section \ref{5.2.1}.

    \item The Functional Module (\(\mathbf{FModule}\)) is the core component that facilitates the customization of the solution. It is packaged as an extension for loading functional expansions and patch updates by the manufacturer, and during the update phase, it is selectively utilized by the user.

    \item The verification chain, denoted as \textsf{Verify}\((\mathbf{H}, C)\), consists of two distinct components.
    The first component, \(\mathbf{H} = (h_1, h_2, \ldots, h_n)\), represents an ordered sequence of Chameleon Hash information, where each \(h_i\) corresponds to a specific Chameleon Hash output, and the order of these elements is strictly preserved. The second component, \(C\), is a unique cryptographic commitment that serves as the anchor of the chain, ensuring the integrity and consistency of the associated Chameleon Hash information. Together, \(\mathbf{H}\) and \(C\) form the verification chain, encapsulating the essential elements required for the verification process.
    
\end{itemize}

\section{IMUP Design} \label{Design}

Aligned with the commonly used two-party update model\cite{rfc9019,9727516}, IMUP models a server and a fleet of IoT devices.  
Figure~\ref{fig:Overview} illustrates its three working phases: P1~\emph{maintenance}, P2~\emph{distribution}, and P3~\emph{verification}.

\noindent$\bullet$\,\textit{P1 Maintenance:} Server-side engineers patch vulnerabilities or add features, then invoke {\small\textsf{IMUP.Package}} to wrap the modified code into a signed functional module \(\mathbf{FModule}\).

\noindent$\bullet$\,\textit{P2 Distribution:} Upon a customisation request, the server searches its cache for a matching firmware; otherwise it assembles the requested module set on demand and signs the result.

\noindent$\bullet$\,\textit{P3 Verification:} Each device—during either offline or OTA update—verifies the received image with \textsf{Verify}\((\mathbf{H},C)\). Only images that pass are installed; failures are discarded immediately.

The remote-only adversary assumed in §3 applies to all three phases; physical attacks remain out of scope.

\subsection{Maintenance Phase}

During maintenance, engineers patch vulnerabilities or add features and then invoke {\small\textsf{IMUP.Package}} to build a self-contained functional module \(\mathbf{FModule}\). Each module bundles the patched code plus a manifest that describes installation steps and dependencies. {\small\textsf{IMUP.Package}} handles scripts and binary objects separately to speed up dependency resolution.

\subsection{Distribution Phase}\label{5.2}

As shown in Fig.~\ref{fig:UpdatePhase}, this phase runs four server-side stages:
S1~\emph{key-module build}, S2~\emph{initialisation}, S3~\emph{security checkpoint}, and S4~\emph{version iteration}.

\noindent$\bullet$\,\textbf{S1.} IMUP packages each task into a code module \(\mathbf{CModule}\).  
A firmware image is an ordered tuple of modules; only an identical order passes verification.  
Because the same module may appear in many sequences, the builder reuses \(\mathbf{CModule}\)s across versions.

\noindent$\bullet$\,\textbf{S2.} The first signed image loaded onto every device yields a trusted pair \textsf{Verify}\((\mathbf{H},C)\), stored locally for later checks.

\noindent$\bullet$\,\textbf{S3--S4.} An S3 build commits to a fresh \(C\) while keeping the hash chain \(\mathbf{H}\) of the previous checkpoint; S4 produces intermediate customised variants that share this verified baseline.  
Any two images between the same S3 checkpoints are therefore security-equivalent and can be swapped at will.

\begin{figure}[htbp]
    \centering
    \includegraphics[width=0.5\textwidth]{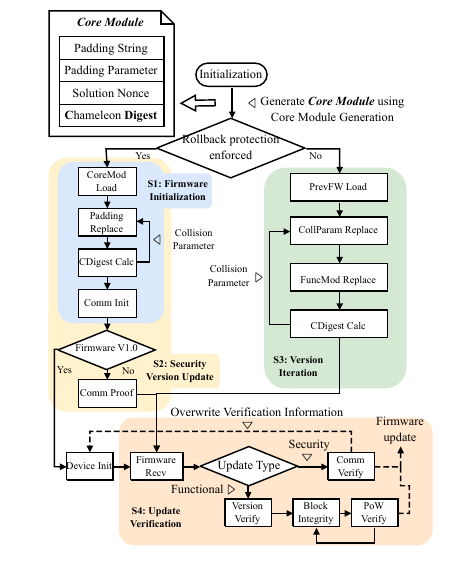}
    \caption{(S1) Key Material Construction, (S2) Firmware Initialization Phase, (S3) Security Version Update Phase, and (S4) Version Iteration Phase collectively belong to the Distribution Phase, while (S5) Devices Verification is part of the Verification Phase.}
    \label{fig:UpdatePhase}
\end{figure}

\subsubsection{Key-Material Construction (S1)}\label{5.2.1}

Stage S1 builds \(N\) code-module templates (\(\mathbf{CModule}\)) and records their digests \(\mathcal{C}\) for later reuse (Alg.~\ref{alg:keygen}).\textcircled{\scriptsize 1} Each \(\mathbf{CModule}\) contains:

\begin{compactitem}
  \item \(PStr\) and \(PParam\): random padding;
  \item a PoW nonce meeting difficulty \(D_{\text{PoW}}\);
  \item a chameleon digest for integrity.
\end{compactitem}

The builder repeats this procedure \(N\) times to obtain \(\mathcal{C}\mathit{Module}\), which feeds the next stage.

\subsubsection{Firmware-Initialisation (S2)}

At S2 the server assembles the first trusted image for each device.

\noindent$\bullet$\textcircled{\scriptsize 2}Select \(L+1\) blocks from \(\mathcal{C}\mathit{Module}\); \(L\) become functional modules and one becomes the commitment block.  

\noindent$\bullet$\textcircled{\scriptsize 3}Replace each \(PStr\) with the corresponding \textit{FModule}; update each \(PParam\) to a fresh collision parameter \(r\) so that every chameleon digest remains valid.  

\noindent$\bullet$\textcircled{\scriptsize 4}After the final block, aggregate the \(L\) digests into \(Block_{\text{Info}}\) and compute the commitment.  

The resulting verification chain \(\mathcal{V}=(\mathbf{H},C)\) is stored on the device and acts as the immutable root for all subsequent updates.

\subsubsection{Security Version Update stage} This stage (summarized in S3 in Figure~\ref{fig:UpdatePhase}) replicates the first four steps of the Firmware Initialization stage and provides the generated firmware with a new verification chain. Beyond the replicated steps, this stage requires the generation of proof of commitment, which is used to verify the prior information commitment. The core concept is that the firmware will enable the device to prove its ability to find a set of hash collisions. These collisions must satisfy the commitment conditions for any given information. This can only be achieved by an entity possessing the trapdoor of the Chameleon Hash function. We only describe the additional step as below.

In \textcircled{\scriptsize 5}, firstly, the information from the previous version will be loaded, which includes the functional module digest \( Block_{\text{Info}} \) and the commitment. Then, the new version block information \( m \) and the proof \( P \) will be calculated, such that 
\[
\textsf{C}_{\text{hash}}(Block_{\text{Info}}, commitment) = \textsf{C}_{\text{hash}}(m, P).
\]
It is noteworthy that for each uniquely generated firmware, corresponding to the same commitment, the proof \( P \) will be completely different.

\subsubsection{Version Iterations stage} This stage (summarized in S4 in Figure~\ref{fig:UpdatePhase}) represents the practical implementation of IMUP, balancing efficiency under customization constraints. Typically, when a user sends a customization request to the server, this process is initiated. During this process, a sequence of \( \mathbf{H} \) will be generated, identical to that in the Verification chain. This implies that updated firmware under the same Verification chain can be substituted and reused. Since users can autonomously load and discard functional modules from an update package based on their needs, the server can flexibly select the number of functionalities to include in a single firmware, depending on the number of modules. This increases the likelihood of reusing generated firmware. For instance, if User\(_1\) requests modules a, b, c, and User\(_2\) requests a, c, d, and there exists an update firmware package containing both of these subsets, then the firmware package can be reused. This approach effectively \textit{reduces the server's load} under high request conditions.

\begin{algorithm}
\caption{Crypto Module Generation}
\label{alg:keygen}
\begin{algorithmic}[1]
\Require Number of Modules $N$, Proof of Work Difficulty $D_\text{PoW}$
\Ensure $\mathcal Crypto Module$
\State $\mathcal Crypto Module$ $\gets \emptyset$

\For{$i = 1$ to $N$}
    \State $PStr, PParam \gets \text{Padding}(32)$
    \State $h \gets C_{\text{hash}}(PStr, PParam)$
    \State $\text{SolutionNonce}, B_\text{hash} \gets \text{PoW}(T, h, D_\text{PoW})$
    \State $\mathbf{CryptoModule}$ $\gets m || r || \text{SolutionNonce} || h$
    \State $\mathcal Crypto Module$.update($\mathbf{CryptoModule}$)
\EndFor
\State \Return $\mathcal Crypto Module$
\end{algorithmic}
\end{algorithm}

Step \textcircled{\scriptsize 6} loads the firmware of the same version as the Verification chain into the server's memory. This step \textcircled{\scriptsize 7} is similar to the Security Version Update stage, the collision parameter \( r \) of the previous block will be replaced in the current firmware to maintain the integrity of the entire firmware. \textcircled{\scriptsize 8} According to the customization request, the required \textit{FModule} will be replaced in the firmware. \textcircled{\scriptsize 9} This process will repeat until each module satisfies the legal requirements of the Chameleon Hash function. Specifically, each unique firmware will generate a unique commitment, which will be stored along with the firmware.

\subsection{Verification Phase}

The firmware verification phase is performed on the IoT device to validate all incoming firmware chains. When the device receives a new firmware chain, the verification chain will be verified. Upon successful verification, the device will select the appropriate \textit{FModule} for updating. This phase consists of steps \textcircled{\scriptsize 10} to \textcircled{\scriptsize 14}, (see in Figure~\ref{fig:UpdatePhase})

\textcircled{\scriptsize 10} The initial device setup must be performed in a trusted environment because this step initializes the verification chain without verification (e.g., when the firmware is first produced at the factory). The device's update interface continuously waits for incoming firmware updates. Verification is divided into two modes: Functional and Security, with each update firmware entering only one branch. 

\textcircled{\scriptsize 11} \textcircled{\scriptsize 12} \textcircled{\scriptsize 13} represent the Functional verification. First, the device checks whether the sequence and values of \(\mathbf{H}\) in the verification chain of the pending firmware update match. Only if they match exactly will the version verification pass. Next, each block is verified for legality, including whether the Chameleon Digest matches and whether the Proof of Work meets the requirements. The Proof of Work module can effectively increase the attacker's cost, reducing the device's attack value. When all verifications pass, the device accepts the firmware update.

\textcircled{\scriptsize 14} represents the Security verification. For convenience, let the original verification chain be denoted as \(\mathcal{V}_{\text{pre}} = (\mathbf{H}_{\text{pre}}, C_{\text{pre}})\), and the updated firmware as \(\mathcal{V}_{\text{Upd}} = (\mathbf{H}_{\text{Upd}}, C_{\text{Upd}})\). If the firmware passes the verification, it satisfies:
\[
\textsf{C}_{\text{hash}}(\mathbf{H}_{\text{pre}}, C_{\text{pre}}) = \textsf{C}_{\text{hash}}(\mathbf{H}_{\text{Upd}}, C_{\text{Upd}})
\]
When the verification passes and each block is completely valid, the device will update the verification chain and accept the updated firmware.

\subsection{Practical Security Mapping}

\begin{table}[htbp]
\centering
\caption{Mapping of Vulnerabilities to IMUP Mitigation Phases}
\resizebox{\linewidth}{!}{
\begin{tabular}{
>{\arraybackslash}c
>{\arraybackslash}c
>{\arraybackslash}p{4.2cm}
>{\arraybackslash}c
}
\toprule
\textbf{Vulnerability} & \textbf{Level} & \textbf{Description} & \textbf{IMUP Mitigation Step(s)} \\
\midrule
CVE-1 & L2 & Config script injection & \textcircled{\scriptsize 2}, \textcircled{\scriptsize 8} \\
CVE-2, 3 & L2 & Rollback marker bypass & \textcircled{\scriptsize 4}, \textcircled{\scriptsize 5} \\
CVE-4 & L3 & Extra module injection & \textcircled{\scriptsize 1}, \textcircled{\scriptsize 3}, \textcircled{\scriptsize 9} \\
CVE-5 & L3 & Stack overflow via header tampering & \textcircled{\scriptsize 8}, \textcircled{\scriptsize 9} \\
\bottomrule
\end{tabular}
}
\label{table:vulnerability-mapping}
\end{table}

To demonstrate how IMUP addresses real-world vulnerabilities, we map the representative flaws identified in Section~\ref{Vulnerabilities} to their corresponding mitigation phases in our design. Table~\ref{table:vulnerability-mapping} summarizes this mapping, highlighting IMUP’s capability to enforce integrity and security under both L2 and L3 customization scenarios.

The steps in Table~\ref{table:vulnerability-mapping} mitigate the listed vulnerabilities by binding all updates to a verifiable integrity structure maintained across customization and version changes. For L2 vulnerabilities such as configuration script injection (CVE-1), steps \textcircled{\scriptsize 2} and \textcircled{\scriptsize 8} ensure that any inserted script must be incorporated into a valid chameleon hash commitment; unauthorized modifications invalidate the commitment and are rejected. For rollback-related flaws (CVE-2 and CVE-3), steps \textcircled{\scriptsize 4} and \textcircled{\scriptsize 5} restrict version changes to those anchored in manufacturer-approved checkpoints, preventing attackers from reverting to unverified or vulnerable firmware. For L3 issues such as extra module injection (CVE-4) and header tampering (CVE-5), steps \textcircled{\scriptsize 1}, \textcircled{\scriptsize 3}, and \textcircled{\scriptsize 9} enforce strict digest ordering and full-chain verification, so any unauthorized module addition or header manipulation leads to immediate failure. Together, these measures ensure that firmware remains integrity-assured under both functional customization and iterative updates.

\section{Security Analysis} 

IMUP resists image forgery and rollback under the adversary model (\ref{ThreatModel}); proofs are provided in Appendix \ref{Appendix:Security Analysis}.

\section{Experiment}

Our experiments on IMUP aim to answer the following Evaluation Questions (EQs):

\begin{enumerate}
    \item \textbf{EQ1:} Can firmware updates remain secure under the constraints of customization rights?
    \item \textbf{EQ2:} Can IoT devices efficiently implement functionalities under the constraints of customization rights?
    \item \textbf{EQ3:} Can servers operate at low cost under the constraints of customization rights?
\end{enumerate}

For \textbf{EQ1} of assessing whether the IMUP scheme can ensure the security of firmware updates when customization rights are allocated, we perform security testing focused on the PoW component. By simulating attacker efforts to forge firmware updates at various security levels, we measure the computational time and resources required to compromise the system. This experiment aims to determine if the increased attack cost introduced by the PoW mechanism effectively deters potential attackers, thereby guaranteeing the security of updates (Section \ref{Q1}). 

For \textbf{EQ2}, we evaluate the efficiency and adaptability of the IMUP scheme on IoT devices with different hardware capabilities. Functional tests are conducted by applying the IMUP scheme to repair various types of vulnerabilities on devices ranging from basic to high-performance. We measure metrics such as verification time and system call overhead to assess whether the IMUP scheme enables IoT devices to efficiently implement required functionalities without imposing significant computational burdens, thus verifying its operational efficiency under customization constraints (Section \ref{Q2}).

For \textbf{EQ3} of examining whether servers can operate stably and cost-effectively when managing firmware updates with customization rights, we design scalability and customization experiments. We simulate high-throughput interactions between the server and users, assessing performance metrics like total processing time, hit rate (the rate of firmware to satisfy different requests), number of generated firmware, and storage consumption under large-scale concurrent requests. By comparing the IMUP scheme with OpenWrt's update strategy, including its package manager, we aim to verify its ability to maintain server operations cheaply. We chose this strategy because it is renowned for its customization capabilities, making it an ideal benchmark for assessing our scheme under extensive customization demands (Section \ref{Q3}).

\subsection{Hardware and Software}

We evaluate IMUP using a testbed with a server, multiple IoT devices, and a unified software setup. Details are provided in Appendix~\ref{Appendix:Hardware}.

\subsection{Attack Overhead}\label{Q1}

\begin{table*}[h!]
    \caption{Attacker Cost Trends with Varying \( D_{\text{PoW}}\) under a 1024-bit Key Size}
    \centering
    \resizebox{1.0\linewidth}{!}{
    \begin{tabular}{c>{\centering\arraybackslash}p{0.1\linewidth}>{\centering\arraybackslash}p{0.1\linewidth}>{\centering\arraybackslash}p{0.1\linewidth}>{\centering\arraybackslash}p{0.1\linewidth}>{\centering\arraybackslash}p{0.1\linewidth}>{\centering\arraybackslash}p{0.1\linewidth}>{\centering\arraybackslash}p{0.1\linewidth}}
    \toprule
    \multirow{2}{*}{} & {\begin{tabular}[c]{@{}c@{}}\textbf{\( D_{\text{c}} \)}\end{tabular}} & 
    \textbf{\( IoT_{\text{vc}} \)}  & \textbf{\( K_\text{99.0\%} \)} & \textbf{\( K_\text{98.5\%} \)} & \textbf{\( K_\text{98.0\%} \)} & \textbf{\( K_\text{97.5\%} \)} & \textbf{\( K_\text{97.0\%} \)} \\ \midrule
    \textbf{\( D_{\text{PoW}}=0 \)} & $<$1ms & 6.37 sec  & 2.05 ms & 65.54 ms & 2.10 sec & 1.12 min & 35.79 min \\ 
    \textbf{\( D_{\text{PoW}}=5 \)} & 0.84 sec & 6.47 sec & 1.73 sec & 55.31 sec & 29.50 min & 15.73 hrs & 8.39 days  \\ 
    \textbf{\( D_{\text{PoW}}=6 \)} & 9.31 sec & 6.49 sec & 19.07 sec & 10.17 min & 5.42 hrs & 2.89 days & 92.56 days  \\ 
    \textbf{\( D_{\text{PoW}}=7 \)} & 2.8 min & 6.61 sec  & 5.73 min & 3.06 hrs & 1.63 days & 52.20 days & 4.57 years \\ 
    \textbf{\( D_{\text{PoW}}=8 \)} & 9.03 min & 6.73 sec  & 18.49 min & 9.86 hrs & 5.26 days & 168.33 days & 14.75 years\\ 
    \bottomrule
    \end{tabular}
    }
    \label{tab:AttackerForgery}
\end{table*}

\begin{table*}[h!]
    \caption{Attacker Cost Trends with Varying Key Sizes under \( D_{\text{PoW}}=7 \)}
    \centering
    \resizebox{1.0\linewidth}{!}{
    \begin{tabular}{c>{\centering\arraybackslash}p{0.1\linewidth}>{\centering\arraybackslash}p{0.1\linewidth}>{\centering\arraybackslash}p{0.1\linewidth}>{\centering\arraybackslash}p{0.1\linewidth}>{\centering\arraybackslash}p{0.1\linewidth}>{\centering\arraybackslash}p{0.1\linewidth}>{\centering\arraybackslash}p{0.1\linewidth}}
    \toprule
    \multirow{2}{*}{} & {\begin{tabular}[c]{@{}c@{}}\textbf{\( D_{\text{c}} \)}\end{tabular}} & 
    \textbf{\( IoT_{\text{vc}} \)} &  \textbf{\( K_\text{99.0\%} \)} & \textbf{\( K_\text{98.5\%} \)} & \textbf{\( K_\text{98.0\%} \)} & \textbf{\( K_\text{97.5\%} \)} & \textbf{\( K_\text{97.0\%} \)} \\ \midrule
    1024-bits & 9.03 min & 6.61 sec & 5.73 min & 3.06 hrs & 1.63 days & 52.20 days & 4.57 years\\ 
    2048-bits & 9.94 min & 7.94 sec & 1.63 days & 4.57 years & \(>\)20 years & \(>\)20 years & \(>\)20 years\\ 
    3072-bits & 10.35 min & 8.17 sec & 4.57 years & \(>\)20 years & \(>\)20 years & \(>\)20 years & \(>\)20 years\\ 
    \bottomrule
    \end{tabular}
    }
    \label{tab:AttackerForgeryKeySize}
\end{table*}

We adopt a non-cooperative game theory model~\cite{von1947theory,Başar2021} to formally analyze the attack overhead of attackers. The participants in the model include the \textbf{Attacker} and the \textbf{Manufacturer} (Defender). The Attacker derives their benefits from forging firmware, while the Defender, as the legitimate firmware publisher, is responsible for maintaining the system's integrity and security. The Defender's strategy involves adjusting the Proof-of-Work difficulty parameter \( D_{\text{PoW}} \) and the key size to maximize system security. Conversely, the Attacker's strategy involves brute-force enumeration to search for possible solutions within the key space.

Within this framework, we design experiments to measure the cost \( A_{\text{C}} \) required for an Attacker to successfully carry out an attack under different \( D_{\text{PoW}} \) and key size settings. As both \( D_{\text{PoW}} \) and the key size increase, the Attacker's cost grows exponentially, thereby reducing the profitability of forging attacks. Additionally, the IMUP mechanism shifts the computational burden to resource-rich server-side environments, ensuring that the verification cost on IoT devices (\( IoT_{\text{vc}} \)) does not increase significantly, while the defender's cost (\( D_{\text{c}} \)) remains manageable. Generally, the computational overhead between the server and the attacker is asymmetry by unique IMUP design when generating and forging firmware, respectively.

\subsubsection{Experimental Setup}
We follow the partial randomness leakage assumption from~\cite{liu2020security}, attacker can recover a certain percentage of the encryption key, denoted as the forged key \( K_p \), where \( p \) represents the percentage of the key exposed. The remaining bytes must be recovered through brute-force enumeration. The attacker employs the encryption function \( \textsf{Encrypt}_{K}(m) \) to attempt to generate legitimate firmware that can pass the device's verification process. Since legitimate firmware must satisfy the Proof-of-Work difficulty parameter \( D_{\text{PoW}} \), each forgery trial increases the computational cost of the attack. Furthermore, considering that computational capabilities limit the Proof-of-Work's computational power, we assume that the attacker possesses a computational power that is 1000 times that of the defender.

In our experiments, we test the impact of different \( D_{\text{PoW}} \) strategies on the attacking cost with a key size of 1024 bits. Additionally, we evaluate the impact of varying key sizes on the Attacker's attack cost under a fixed \( D_{\text{PoW}} = 8 \) strategy. Furthermore, we assess the performance of each strategy on a 400MHz IoT device to show that security parameters of key size and PoW have almost no overhead increase to it, ensuring the practicality of the IMUP design.

\subsubsection{Results and Analysis}
Table~\ref{tab:AttackerForgery} presents the attack cost with varying Proof-of-Work difficulty parameters (\( D_{\text{PoW}} \)) under a fixed key size of 1024 bits. \( D_{\text{PoW}} \) refers to the difficulty coefficient of Proof-of-Work, where a higher value signifies greater computational effort required to complete the PoW task.

\noindent{\bf Results without PoW.} When \( D_{\text{PoW}} = 0 \), the attack cost to forge firmware updates is extremely low. In the case of \( K_{\text{97.5\%}} \), it only takes 35.79 minutes to successfully forge the firmware. This indicates that without the PoW module, an Attacker can quickly generate forged firmware with minimal computational effort when a portion of the key is exposed, posing significant security risks.

\noindent{\bf Results with PoW.} As the PoW difficulty level \( D_{\text{PoW}} \) increases, the attacker cost escalates substantially. For example, at \( D_{\text{PoW}} = 5 \), the forgery time ranges from 1.73 seconds (\( K_{\text{99.5\%}} \)) to 8.39 days (\( K_{\text{97.5\%}} \)). When \( D_{\text{PoW}} \) increases to 8, the cost grows exponentially, ranging from 18.49 minutes to 14.75 years, even though the attacker possesses 1000 times the computational power of the defender. However, the defense cost remains low at only 9 minutes to the server. This asymmetric computational overhead design effectively prevents the forgery attack. Even when a vast majority of key bits are exposed, the system still provides sufficient security time for the defender to respond and mitigate losses.

Table~\ref{tab:AttackerForgeryKeySize} explores the attack cost \( A_{\text{C}} \) with varying key sizes under a fixed \( D_{\text{PoW}} = 8 \). The results demonstrate that as the key size increases from 1024 bits to 3072 bits, the attack cost rises significantly. Specifically, for a 1024-bit key, the cost reaches 168.33 days, whereas for a 2048-bit key, it exceeds 20 years. For a 3072-bit key, the attack cost remains above 20 years across all measured parameters. This substantial increase underscores the robustness of the system against attacks, even in scenarios involving catastrophic key leakage.

Another prominent finding is that the computational overhead on IoT devices (\( IoT_{\text{vc}} \)) remains largely unaffected across different \( D_{\text{PoW}} \) and key sizes. For instance, in Table~\ref{tab:AttackerForgery}, \( IoT_{\text{vc}} \) shows minimal fluctuations, staying within approximately 6.37 to 6.73 seconds regardless of \( D_{\text{PoW}} \). Similarly, in Table~\ref{tab:AttackerForgeryKeySize}, \( IoT_{\text{vc}} \) remains consistent, further validating that the IMUP mechanism effectively offloads the computational burden to server-side. This ensures that IoT devices maintain their operational efficiency without compromising security.

\subsection{Overhead on IoT Device}\label{Q2}

\begin{figure*}[htbp]
    \centering
    \includegraphics[width=\textwidth, trim=0cm 0.35cm 0cm 0cm, clip]{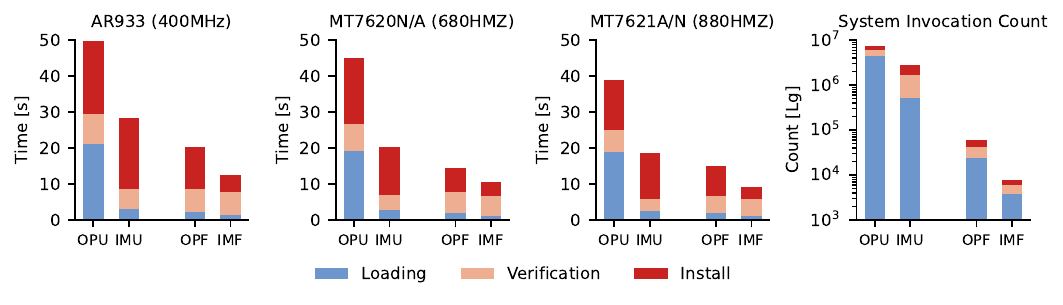}
    \caption{The time consumption and system call distributions of OPKG and IMUP for functionality updates (OPU, IMU) and security fixes (OPF, IMF) across different CPU frequencies.}
    \label{fig:IoTCost}
\end{figure*}

To evaluate the efficiency and applicability of our IMUP scheme across a wide range of IoT devices with diverse hardware capabilities, we measure both the execution time of firmware patching and the number of system calls invoked during the update process. As a baseline, we adopt the widely used OPKG package manager due to its extensive support for customization \cite{openwrt_opkg,openmoko_opkg}. 

OPKG is a general-purpose solution for installing and managing software components, enabling users to customize device functionalities through an open repository system. Despite its strengths, it is \emph{not} designed to address IMUP’s core challenges—particularly the seamless integration of user-friendly customization with vendor-side security. Under the OPKG model, each device independently fetches updates from external repositories, increasing computational overhead and complicating maintenance. 

IMUP simplifies customization for both users and vendors by modularizing firmware and ensuring integrity through cryptographic mechanisms. This reduces risks and complexities in decentralized updates, easing the user burden and enabling vendors to maintain secure, consistent firmware at scale.

\subsubsection{Experimental Setup}

The parameters of the IMUP scheme are set to $L=7$, $D_{\mathrm{PoW}}=5$, and the cryptographic keys employed have a 2048-bit key size. We plan to conduct two categories of experiments: \textbf{Functionality Update} and \textbf{Vulnerability Fix}. 
The specific implementation details are provided in Appendix~\ref{Appendix:Details}.

\smallskip
\subsubsection{Results and Analysis}

Figure~\ref{fig:IoTCost} presents a comparative analysis of update operations across three IoT devices with different computation resources (CPU frequencies ranging from AR933 (400,MHz) and MT7620N/A (680,MHz) to MT7621A/N (880,MHz). 
We evaluate two update paradigms: OPKG-based (OPU for functional updates, OPF for security fixes) and IMUP-based (IMU for functional updates, IMF for security fixes). The metrics considered include execution time during loading, verification, and installation stages, as well as the associated system call counts.

\noindent{\bf Reduced Overhead on Low-Resource Devices.} On the AR933 IoT device (400,MHz), the baseline OPU approach costs 21.24 s for loading and 20.16 ms for installation, whereas our IMU shortens these durations to 3.16 s and 19.62 s, respectively—an 85\% improvement in the loading stage. This acceleration is primarily achieved by offloading complex dependency resolution and remote component fetching to the server side, eliminating the need for frequent on-device lookups. As a result, the loading phase’s system calls decreased from over 9 million under OPU to about 31,849 under our IMU. A similar trend holds for security patching: OPF’s loading time decreases from 2.34 s to 1.51 s under our IMF. These results demonstrate IMUP's superior efficacy in alleviating computational and networking overhead, particularly for IoT devices with limited processing capabilities.

\noindent{\bf Consistent Gains with Increasing Computation Resource.} As computation resources increase e.g., CPU frequency increases, both OPKG and IMUP approaches improve their performance. However, IMUP consistently maintains a substantial advantage. For instance, at 880,MHz (MT7621A/N), IMUP loading times for functional updates (IMU) and security patches (IMF) still outperform OPU and OPF by wide margins. IMU reduces loading from 18.98 s (OPU) to 2.51 s, and IMF cuts it from 1.89 s (OPF) to 1.01 s. This trend demonstrates that IMUP’s efficiency is feasible at the low end of hardware capabilities and scales effectively as IoT device performance improves.

\begin{table*}[h!]
    \caption{Summary of Server Performance under Large-Scale Requests (detailed results are in Appendix)}
    \centering
    \resizebox{1.0\linewidth}{!}{
    \begin{tabular}{>{\centering\arraybackslash}p{0.07\linewidth}>{\centering\arraybackslash}p{0.07\linewidth}c>{\centering\arraybackslash}p{0.1\linewidth}ccc}
        \toprule
        \textbf{Modules Number} & \textbf{Requests Times} & \textbf{Scheme Type} & 
        \textbf{Total Time (s)} & \textbf{Hit Rate (\%)} & 
        \textbf{Number of Firmware} & \textbf{Storage (GB)} \\ 
        \midrule
        \multirow{3}{*}{200} & \multirow{3}{*}{10,000} & Monolithic Rebuild & 282,313.12 & 1.77 & 9,823 & 356.85 \\ 
        & & Package Manager & 332.13 & N/A & N/A & N/A \\ 
        & & IMUP \(L=7\) & 230.88 & 0.64 & 3,608 & 7.11 \\ 
        \cmidrule(lr){2-7}
        & \multirow{3}{*}{20,000} & Monolithic Rebuild & 576,066.80 & 0.3 & 19,940 & 724.38 \\ 
        & & Package Manager & 667.70 & N/A & N/A & N/A \\ 
        & & IMUP \(L=7\) & 318.65 & 75.39 & 4,923 & 9.69 \\ 
        \cmidrule(lr){2-7}
        & \multirow{3}{*}{30,000} & Monolithic Rebuild & 926,320.50 & 0.07 & 29,978 & 1,089.04 \\ 
        & & Package Manager & 1,014.21 & N/A & N/A & N/A \\ 
        & & IMUP \(L=7\) & 357.55 & 81.88 & 5,437 & 10.71 \\ 
        \midrule
        \multirow{1}{*}{2,000} & 30,000 & Monolithic Rebuild & 816,001.46 & 0.06 & 29,981 & 1,089.15 \\ 
        & & Package Manager & 1,175.53 & N/A & N/A & N/A \\ 
        & & IMUP \(L=7\) & 1,080.89 & 45.71 & 16,286 & 32.06 \\ 
        \midrule
        \multirow{1}{*}{4,000} & 30,000 & Monolithic Rebuild & 712,770.42 & 0.07 & 29,978 & 1,089.04 \\ 
        & & Package Manager & 1,137.92 & N/A & N/A & N/A \\ 
        & & IMUP \(L=7\) & 389.70 & 45.13 & 16,460 & 32.40 \\ 
        \bottomrule
    \end{tabular}
    }
    \label{tab:ServerPerformance}
\end{table*}

\subsection{Server Stability and Cost Efficiency}\label{Q3}

To examine whether servers can operate at low cost under the constraints of customization rights (\textbf{RQ3}), we conduct large-scale performance evaluations focusing on four quantifiable metrics:
\begin{itemize}
    \item \textbf{Total Time}: The cumulative time required by the server to retrieve existing firmware images and generate new ones.
    \item \textbf{Hit rate}: The percentage of requests fulfilled using pre-existing images in the server cache, reflecting the server’s ability to leverage previously generated firmware. Formally,
    \[
    \text{Hit rate} = \frac{\text{Number of Reused}}{\text{Total Number of Requests}} \times 100\%
    \]
    \item \textbf{Number of Generated Firmware}: The total count of distinct firmware images produced and stored by the server, indicating how effectively redundant image builds are avoided.
    \item \textbf{Storage Consumption}: The total storage space occupied by all retained firmware images, serving as an indicator of ongoing operational costs.
\end{itemize}

We evaluate these metrics under varying workload intensities by issuing 10,000, 20,000, and 30,000 concurrent image build requests. Following the modularity documented by OpenWrt as of August 2024~\cite{openwrt2024}, we vary the number of modular firmware packages available on the server (200, 2,000, and 4,000) to represent different levels of functional diversity. A request generator, based on module popularity and allowing duplicate requests, is employed to reflect the dynamics observed in real-world scenarios. For each incoming request, the server either retrieves an existing cached image or compiles a new one, depending on the availability of a suitable pre-generated image.

In addition, we explore how adjustments to chain lengths (\(L\)), security parameters (\(D_{\text{PoW}}\)), and cryptographic key sizes (e.g., 2048-bit keys) influence on these metrics (Appendix \ref{Appendix:Configuration}). Through these comprehensive evaluations, we demonstrate that it is possible for servers to maintain minimum processing times, high hit rates, and low storage consumption when handling large-scale, customization-intensive request workloads.

\subsubsection{Experimental Setup}

We consider two baseline update strategies for comparison, aligning with the analysis methodology presented in Ebbers \textit{et al.}'s study~\cite{ebbers2022large}.

\begin{enumerate}
    \item \textbf{Monolithic Rebuild Strategy}: Upon receiving a user request, the server searches for an existing firmware image. If no pre-compiled image is available, the server recompiles the entire kernel. This strategy incurs high time and memory costs due to complete recompilation.
        
    \item \textbf{Package Manager Strategy}: IoT devices send customized requests to a server, which then returns the corresponding packages. Upon receiving each package, the IoT devices verify it and resolve relevant dependencies, recursively requesting any missing dependencies from the server until the update process is complete. We refer to this entire sequence of interactions as a single “request.” To more accurately assess server performance, the processing time on the IoT devices is excluded from our measurements. Since the image of this strategy is not reusable, we do not consider metrics such as hit rate and storage size.
\end{enumerate}

\subsubsection{Throughput Analysis}

Table~\ref{tab:ServerPerformance} presents the server performance evaluation under large-scale requests, comparing the Monolithic Rebuild, Package Manager, and IMUP (\( L=7 \)) schemes across varying numbers of modules and request volumes.

The results indicate that the IMUP scheme consistently exhibits lower server-side costs compared to the Monolithic Rebuild and Package Manager strategies. This efficiency is attributable to the rational allocation of customization rights between users and manufacturers, which effectively increases the firmware hit rate. 
By allowing manufacturers to generate and store a wider range of firmware variants, the likelihood that a requested firmware already exists on the server increases, reducing the need for time-consuming on-the-fly generation.

Although the hit rate decreases slightly as the number of modules increases, it remains above 40\% even at higher module counts. This demonstrates the scalability of the IMUP scheme in handling large-scale firmware requests while maintaining efficient server performance.
In practical applications, more popular components and modules are requested more frequently, leading to higher hit rates than those observed in the experimental setup. The popularity of certain modules means that they are likely to be cached on the server, further enhancing the efficiency of the IMUP scheme in real-world scenarios.
For example, when the number of modules is 200 and the number of requests is 20,000, the IMUP scheme achieves a hit rate of 75.39\%, significantly higher than the Monolithic Rebuild strategy's hit rate of only 0.30\%. This high hit rate contributes to the IMUP scheme's lower total processing time and reduced storage requirements.

Overall, the IMUP scheme's ability to maintain a high firmware hit rate, even as the number of modules increases, results in lower server processing times and storage consumption. This makes it a cost-effective solution for firmware distribution across deployments requiring different combinations of modules and varying update granularities.

\smallskip
\section{Conclusion}
We present the first comprehensive study of firmware customization under strict security and efficiency demands. In addressing our research questions—ensuring robust security while supporting modular user updates—we propose IMUP. By combining a Chameleon Hash function with proof-of-work, IMUP shifts cryptographic overhead to powerful servers, dramatically raising the cost of firmware forgery. Its modular design meets user customization needs and simplifies version control. Evaluations on diverse IoT devices confirm IMUP’s effectiveness in balancing security and performance for next-generation firmware updates.

\newpage

\bibliographystyle{IEEEtran}
\bibliography{references}

\appendix

\subsection{Threat-Model Details} \label{Appendix:Threat-Model}

Through our investigation, we note that users primarily obtain update components through two methods: Offline and Over-The-Air (OTA) Update, as shown in Figure \ref{Threat Models}. These update methods form the core of our unified threat model, outlining potential adversarial actions and vulnerabilities within the IoT update ecosystem.

\noindent\textbf{Offline Update.} This mode is commonly used for customization. Users typically search for suitable customized update configurations and firmware on community forums or official websites. After downloading the relevant firmware, they manually load it onto their IoT devices. This process relies on the user's technical expertise and security awareness. In general, offline updates do not occur automatically, even when the current firmware version may be subject to security vulnerabilities.

\noindent\textbf{OTA Update.} It is the predominant update solution offered by commercial firmware providers. When manufacturers sync firmware to servers, IoT devices download and update the designated firmware at appropriate intervals. This fully automated update process reduces user involvement and enhances user experience. However, it also means that users and devices cannot customize functionalities based on the current operating environment, significantly diminishing both user experience and device performance. 

\noindent\textbf{Adversary Model.}  
We assume that adversaries possess advanced technical skills—including proficiency in reverse-engineering firmware updates and binary code—which enables them to extract complete firmware images, log files, and internal details (e.g., memory layouts, update verification logic, and embedded symmetric keys) from targeted IoT devices. Motivated by financial gain~\cite{cybersecuritynews_chinese_hijackers,secrss_indian_power_grid}, political objectives, or competitive sabotage, these adversaries employ techniques such as man-in-the-middle (MITM) attacks and exploitation of weak signature verification to intercept and tamper with firmware updates.

Although adversaries may have substantial computational resources, their capabilities are constrained by factors such as network latency, legal limitations, and restricted physical access. We assume that the vendor’s signing infrastructure is secure and that IoT devices implement a secure boot process. Aligned with~\cite{ammar2024bridging,lorych2024hardware}, physical and side-channel attacks are out of scope. These attack vectors require expensive laboratory equipment, prolonged physical access to the device, and fine-grained manipulation capabilities, which are not practical for scalable or remote attacks. As demonstrated by Wu {\it et al.}~\cite{wu2024your}, update vulnerabilities can be remotely exploited without requiring device teardown or advanced hardware probing. Therefore, we focus on adversaries capable of launching remote, scalable threats—such as firmware tampering and protocol-level manipulation—which are both more prevalent and feasible in real-world IoT deployment environments. Furthermore, physical protection mechanisms, such as tamper-resistant packaging or secure enclaves, are often employed as complementary defenses at the hardware level~\cite{10014854}, while our study aims to strengthen the integrity and resilience of the firmware update pipeline at the software and system level.

\begin{table*}[h!]
    \caption{Performance Evaluation of Monolithic Rebuild, Package Manager, and IMUP Schemes 1024-bit Key Size)}
    \centering
    \resizebox{1.0\linewidth}{!}{
    \begin{tabular}{lcc>{\centering\arraybackslash}p{0.15\linewidth}>{\centering\arraybackslash}p{0.15\linewidth}}
        \toprule
        \textbf{Metric} & \textbf{Monolithic Rebuild} & \textbf{Package Manager} & 
        \begin{tabular}[c]{@{}c@{}}\textbf{IMUP}\\ (\( L=7 \), \( D_{\text{PoW}}=5 \))\end{tabular} & 
        \begin{tabular}[c]{@{}c@{}}\textbf{IMUP}\\ (\( L=7 \), \( D_{\text{PoW}}=6 \))\end{tabular} \\ 
        \midrule
        Preparation Time (s) & N/A & N/A & 8.83 & 195.47 \\ 
        First Processing Time (s) & 157.56 & 2.01 & 8.89 & 195.52 \\ 
        Subsequent Processing Time (s) & N/A & 1.96 & 0.06 & 0.05 \\ 
        Memory Usage (MB) & 342.37 & 48.60 & 20.24 & 30.44 \\
        Avg. Time (5 Images, s) & 157.56 & 1.97 & \textbf{1.81} & 39.14
 \\     Avg. Time (110 Images, s) & N/A & 1.96 & 0.14 & \textbf{1.83}
 \\
        \bottomrule
    \end{tabular}
    }
    \label{tab:PerformanceEvaluation}
\end{table*}

\begin{table*}[h!]
    \caption{IMUP Configurations (\( D_{\text{PoW}}=5 \)) Under High-Volume Requests for \( L=6 \), \( L=7 \), and \( L=8 \) (1024-bit Key Size)}
    \centering
    \resizebox{1.0\linewidth}{!}{
    \begin{tabular}{>{\raggedright\arraybackslash}p{0.25\linewidth}>{\centering\arraybackslash}p{0.05\linewidth}>{\centering\arraybackslash}p{0.05\linewidth}>{\centering\arraybackslash}p{0.05\linewidth}>{\centering\arraybackslash}p{0.05\linewidth}>{\centering\arraybackslash}p{0.05\linewidth}>{\centering\arraybackslash}p{0.05\linewidth}>{\centering\arraybackslash}p{0.05\linewidth}>{\centering\arraybackslash}p{0.05\linewidth}>{\centering\arraybackslash}p{0.05\linewidth}}
        \toprule
        \textbf{Requests} & \multicolumn{3}{c}{\textbf{10,000}} & \multicolumn{3}{c}{\textbf{20,000}} & \multicolumn{3}{c}{\textbf{30,000}} \\ 
        \cmidrule(lr){2-4} \cmidrule(lr){5-7} \cmidrule(lr){8-10}
        \textbf{Scheme} & \textbf{\( L=8 \)} & \textbf{\( L=7 \)} & \textbf{\( L=6 \)} & \textbf{\( L=8 \)} & \textbf{\( L=7 \)} & \textbf{\( L=6 \)} & \textbf{\( L=8 \)} & \textbf{\( L=7 \)} & \textbf{\( L=6 \)} \\ 
        \midrule
        Search Time (s) & 3.94 & 4.13 & 4.71 & 8.45 & 9.67 & 11.99 & 12.29 & 14.62 & 18.96 \\ 
        Average Processing Time (ms) & 0.07 & 0.06 & 0.05 & 0.07 & 0.06 & 0.05 & 0.07 & 0.06 & 0.05 \\ 
        Average Search Time (ms) &0.39 & 0.41 & 0.47 & 0.42 & 0.48 & 0.60 & 0.41 & 0.49 & 0.63 \\ 
        Hit Rate (\%) & \textbf{67.83} & 63.92 & 59.90 & \textbf{79.73} & 75.39 & 70.45 & \textbf{85.14} & 81.88 & 77.32 \\ 
        Number of Generated Firmware & \textbf{3217} & 3,608 & 4,010 & \textbf{4055} & 4,923 & 5,911 & \textbf{4460} & 5,437 & 6,805 \\ 
        Storage (GB) & 7.24 & 7.11 & \textbf{6.77} & \textbf{9.13} & 9.69 & 9.98 & \textbf{10.04} & 10.71 & 11.49 \\ 
        \bottomrule
    \end{tabular}
    }
    \label{tab:ComparisonOfIMUP}
\end{table*}

\subsection{Security Analysis} \label{Appendix:Security Analysis}

We now analyze the security of IMUP based on the scheme described in Section \ref{Design}, and provide relevant security proofs using the attacker's capabilities and objectives as described in the threat model of Section \ref{ThreatModel}. Aligned with~\cite{wrotniak2019provable,li2024secure,pasquini2024breach,liu2024enforcing}, we do not take Side Channel Attacks (SCAs) into account for the following reasons. First, attackers cannot directly access the victim's device, making traditional SCAs~\cite{299754,horvath2024sok,294617} that rely on physical phenomena (such as timing delays, power consumption, electromagnetic leaks, etc.) infeasible in this scenario. Second, the core trapdoor information is used only on the server side, and we assume the server environment is secure and resistant to remote SCAs. Even if attackers obtain a copy of the device, they can only extract information through SCAs. However, within our defined threat model and assumptions, all information except for the trapdoor data stored on the server is public. Therefore, SCAs are excluded from our consideration. 

\begin{theorem}
If the selected chameleon hash function possesses properties like forgery resistance and collision resistance, then the scheme is a secure firmware defense mechanism.
\end{theorem}

\begin{proof}
The IoT device accepts the firmware only if the verification algorithm outputs bit 1. The attacker attempts to tamper with the firmware to inject malicious data into the IoT device.

\textbf{Tampering}: Suppose there exists a polynomial-time Tamper algorithm that takes a tampered firmware chain $FakeFirmWare$ as input and outputs a firmware that passes the verification algorithm, denoted as:
\begin{center}
    FirmWare $\gets$ \text{Tamper}(\textbf{CModule}, FakeFirmWare)
\end{center}

At this stage, the attacker successfully forges the firmware. The Tamper algorithm must recalculate the hash value of the tampered parts and assign new random parameters $R$ to maintain chain integrity. Specifically, the Tamper algorithm must find a pair of chameleon hash values in polynomial time that satisfy the following equation:

\[
\textsf{C}_{\text{hash}}(\textbf{CModule}, r) = \textsf{C}_{\text{hash}}( \text{FakeFirmWare},R)
\]

Because the underlying chameleon hash function possesses weak collision resistance, no algorithm can find a pair of chameleon hash collisions in polynomial time. Thus, the scheme is tamper-resistant.
    
\textbf{Forgery}: Suppose there exists a polynomial-time Forgery algorithm that takes a commitment value as input and outputs firmware that passes the verification algorithm but cannot be used normally, denoted as:
{\setlength{\abovedisplayskip}{4pt}
\setlength{\belowdisplayskip}{4pt}
\begin{center}
$\text{FirmWare} \gets \text{Forgery}(Random, C)$
\end{center}
}
At this stage, the attacker successfully forges the firmware. This forged firmware damages the IoT device firmware, disrupting the device's normal operation. The Forgery algorithm must generate an output that passes the HVerify algorithm based on the commitment value, denoted as: 
{\setlength{\abovedisplayskip}{4pt}
\setlength{\belowdisplayskip}{4pt}
    \begin{center}
    $b = 1 \gets \text{HVerify}(C)$
    \end{center}
}
Specifically, the Forgery algorithm must find a pair of chameleon hash values in polynomial time that satisfy the following equation:

\[
\textsf{C}_{\text{hash}}(\text{FirmWare}, r) = \textsf{C}_{\text{hash}}( \text{FakeFirmWare},C)
\]

Since the underlying chameleon hash function has weak collision resistance, there is no algorithm within polynomial time that can find a pair of chameleon hash collisions without knowing the trapdoor. Therefore, the scheme is forgery-resistant.
\end{proof}



\subsection{IMUP Generation Cost} \label{Appendix:Generation}

For the IMUP scheme, we introduce several unique metrics to characterize its distinctive features:

\begin{itemize}
    \item \textbf{Preparation Time (s)}: The time required to generate the cryptographic modules that meet the scheme's requirements.
    \item \textbf{First Processing Time (s)}: The time from the start of the preparation phase to the complete generation of the initial functional firmware image.
    \item \textbf{Subsequent Processing Time (s)}: The time needed to generate additional firmware images after the first image has been successfully created.
    \item \textbf{Search Time (s)}: The average time spent searching for usable firmware among the already generated images.
\end{itemize}

\begin{table*}[h!]
    \caption{Comparison of Server Performance under Large-Scale Requests for Different Module Numbers and Schemes}
    \centering
    \resizebox{1.0\linewidth}{!}{
    \begin{tabular}{>{\centering\arraybackslash}p{0.07\linewidth}>{\centering\arraybackslash}p{0.07\linewidth}c>{\centering\arraybackslash}p{0.1\linewidth}ccc}
        \toprule
        \textbf{Modules Number} & \textbf{Requests Times} & \textbf{Scheme Type} & 
        \textbf{Total Processing Time (s)} & \textbf{Hit Rate (\%)} & 
        \textbf{Number of Firmware} & \textbf{Storage (GB)} \\ 
        \midrule
        \multirow{3}{*}{200} & \multirow{3}{*}{10,000} & Monolithic Rebuild & 282,313.12 & 1.77 & 9,823 & 356.85 \\ 
        & & Package Manager & 332.13 & N/A & N/A & N/A \\ 
        & & IMUP \(L=7\) & 230.88 & 0.64 & 3,608 & 7.11 \\ 
        \cmidrule(lr){2-7}
        & \multirow{3}{*}{20,000} & Monolithic Rebuild & 576,066.80 & 0.3 & 19,940 & 724.38 \\ 
        & & Package Manager & 667.70 & N/A & N/A & N/A \\ 
        & &  IMUP \(L=7\) & 318.65 & 75.39 & 4,923 & 9.69 \\ 
        \cmidrule(lr){2-7}
        & \multirow{3}{*}{30,000} & Monolithic Rebuild & 926,320.50 & 0.07 & 29,978 & 1,089.04 \\ 
        & & Package Manager & 1,014.21 & N/A & N/A & N/A \\ 
        & &  IMUP \(L=7\) & 357.55 & 81.88 & 5,437 & 10.71 \\ 
        \midrule
        \multirow{3}{*}{2,000} & \multirow{3}{*}{10,000} & Monolithic Rebuild & 264,132.78 & 0.44 & 9,956 & 361.68 \\ 
        & & Package Manager & 353.05 & N/A & N/A & N/A \\ 
        & &  IMUP \(L=7\) & 355.97 & 43.98 & 5,602 & 11.03 \\ 
        \cmidrule(lr){2-7}
        & \multirow{3}{*}{20,000} & Monolithic Rebuild & 635,164.40 & 0.1 & 19,980 & 725.84 \\ 
        & & Package Manager & 803.16 & N/A & N/A & N/A \\ 
        & &  IMUP \(L=7\) & 706.57 & 45.12 & 10,977 & 21.61 \\ 
        \cmidrule(lr){2-7}
        & \multirow{3}{*}{30,000} & Monolithic Rebuild & 816,001.46 & 0.06 & 29,981 & 1,089.15 \\ 
        & & Package Manager & 1,175.53 & N/A & N/A & N/A \\ 
        & &  IMUP \(L=7\) & 1,080.89 & 45.71 & 16,286 & 32.06 \\ 
        \midrule
        \multirow{3}{*}{4,000} & \multirow{3}{*}{10,000} & Monolithic Rebuild & 260,604.29 & 1.77 & 9,823 & 356.85 \\ 
        & & Package Manager & 389.80 & N/A & N/A & N/A \\ 
        & &  IMUP \(L=7\) & 368.64 & 42.08 & 5,792 & 11.41 \\ 
        \cmidrule(lr){2-7}
        & \multirow{3}{*}{20,000} & Monolithic Rebuild & 633,892.80 & 0.3 & 19,940 & 724.38 \\ 
        & & Package Manager & 792.38 & N/A & N/A & N/A \\ 
        & &  IMUP \(L=7\) & 714.49 & 44.66 & 11,068 & 21.79 \\ 
        \cmidrule(lr){2-7}
        & \multirow{3}{*}{30,000} & Monolithic Rebuild & 712,770.42 & 0.07 & 29,978 & 1,089.04 \\ 
        & & Package Manager & 1,137.92 & N/A & N/A & N/A \\ 
        & &  IMUP \(L=7\) & 389.70 & 45.13 & 16,460 & 32.40 \\ 
        \bottomrule
    \end{tabular}
    }
    \label{tab:AppendixServerPerformance}
\end{table*}

Table~\ref{tab:PerformanceEvaluation} presents the efficiency evaluation of firmware generation for the Monolithic Rebuild, Package Manager and IMUP schemes. In this section, the IMUP scheme uses a 1024-bit key size and is evaluated under two different security parameters (\( D_{\text{PoW}} = 5 \) and \( D_{\text{PoW}} = 6 \)).

Results show that the IMUP scheme significantly outperforms the Monolithic Rebuild strategy and is also superior to Package Manager in terms of processing time and memory usage. 
Although the initial generation time of IMUP is slower than that of the Package Manager strategy due to the generation of cryptographic modules and the PoW computation, the IMUP configurations surpass the Package Manager in resource consumption after generating as few as 5 firmware images. Specifically, the IMUP scheme with \( D_{\text{PoW}}=5 \) achieves a lower average processing time per image compared to the Package Manager when generating 5 images, and both IMUP configurations significantly outperform the Package Manager after generating 110 images. This indicates that the IMUP scheme becomes more efficient and resource-effective than the Package Manager strategy as the number of firmware images scales up, highlighting its suitability for large-scale firmwares.


\subsection{IMUP Configuration Analysis} \label{Appendix:Configuration}

Table \ref{tab:ComparisonOfIMUP} benchmarks IMUP with chain lengths \(L\!\in\!\{6,7,8\}\) under a fixed PoW difficulty (\(D_{\text{PoW}}{=}5\)) and a 1024-bit key.  
Besides storage, hit rate, and the number of generated images, we track **server-side “search time’’—the time to locate a matching image in the cluster.**

\noindent\textbf{Findings.} (1) A longer chain raises the hit rate, so fewer new images are built and search time drops. (2) Publishers can tune \(L\) to match their module count and request volume, balancing reuse against storage. Overall, high hit rates and short search times show that IMUP scales well for large customisation bursts while keeping operational cost low; its tunable \(L\) and security level let vendors optimise for their own environments.

\subsection{Hardware and Software} \label{Appendix:Hardware}

The server was equipped with an AMD Ryzen 7 4800HS processor (2.90 GHz) with Radeon Graphics, 16 GB of RAM, and a 512 GB SSD, running Ubuntu 20.04 LTS. This server acted as the firmware publisher, responsible for generating and distributing firmware updates.

The IoT devices selected for testing represented a range of hardware capabilities commonly found in the field\cite{en16083465}. We utilized three router models: (1) Atheros AR9330 with a MIPS 24Kc CPU at 400 MHz and 32 MB RAM (basic performance), (2) MT7620N/A with a MIPS24KEc CPU at 580 MHz and 128 MB RAM (standard performance), and (3) MT7621A/N featuring a MIPS1004Kc dual-core CPU at 880 MHz and 256 MB RAM (high performance). These devices operated as firmware recipients, allowing us to assess the IMUP scheme's efficiency and scalability across different hardware configurations.

For the software configuration, all IoT devices ran the OpenWrt 19.07 operating system\cite{openwrt_targets_19_07_0}. The server environment included the OpenWrt build system and utilized tools such as IMUPackage for modular firmware packaging and custom scripts for generating and verifying firmware packages. The Opkg package manager was employed for baseline comparisons. Devices were connected via a secure, isolated local area network to ensure controlled testing conditions.

\subsection{Functionality Update \& Vulnerability Fix Details} \label{Appendix:Details}

We selected five key functionalities as test subjects: 
\begin{enumerate}[label=(\Alph*)]
    \item Network load balancing
    \item Traffic control
    \item VPN configuration
    \item File sharing services
    \item Ad blocking and content filtering
\end{enumerate}
These collectively encompass critical aspects of security, efficiency, and user experience. 
Users send update requests to the server to initiate the update process, upon which the server generates and returns firmware update packages containing the selected functionalities. 
Once the packages are received, the devices execute the corresponding functionality updates.

\subsubsection{Vulnerability Fix}
For the vulnerability fix experiments, four known security vulnerabilities were selected: CVE-2019-19945, CVE-2021-28961, CVE-2023-24181 and CVE-2023-24182.
The server generates update packages containing the respective vulnerability patches and transmits them to the target devices. 
After receiving these patches, the devices apply them, completing the vulnerability remediation process.

\subsection{Server Stability and Cost Efficiency} \label{Appendix:ServerStability}

Table \ref{tab:AppendixServerPerformance} provides a comprehensive comparison of server performance under large-scale requests for different module numbers and schemes. The table illustrates the impact of varying request volumes (10,000, 20,000, and 30,000) and module numbers (200, 2,000, and 4,000) on server processing time, firmware hit rates, number of generated firmware, and storage consumption.

For Modules Number = 200, the data covers all three request levels, showcasing how the IMUP scheme significantly improves server efficiency, reducing processing time and storage requirements compared to Monolithic Rebuild and Package Manager strategies.

For Modules Number = 2,000 and 4,000, the table focuses on the largest request volume (30,000), demonstrating the scalability of IMUP in handling more complex and diverse customization scenarios while maintaining lower computational costs and higher hit rates than alternative schemes.

The detailed results highlight the superiority of the IMUP framework in balancing efficiency and flexibility for large-scale firmware update processes.


\end{document}